\documentclass[sigconf]{aamas} 

\usepackage{balance} 
\usepackage{graphicx}
\usepackage{natbib}
\usepackage{doi}
\usepackage{graphicx}
\usepackage{amsmath}
\usepackage{xcolor}
\usepackage{amsthm}
\usepackage{booktabs}
\usepackage{algorithm}
\usepackage{enumitem}
\usepackage{algorithmic}

\newtheorem{theorem}{Theorem}[section]
\newtheorem{lemma}[theorem]{Lemma}

\newcommand{\eat}[1]{}

\acmSubmissionID{???}

\title[Analysis of a Learning Based Algorithm for Budget Pacing]{Analysis of a Learning Based Algorithm for Budget Pacing}

\author{MohammadTaghi Hajiaghayi}
\affiliation{
  \institution{University of Maryland}
  \city{College Park}
  \country{USA}}
\email{hajiagha@umd.edu}

\author{Max Springer}
\affiliation{
  \institution{University of Maryland}
  \city{College Park}
  \country{USA}}
\email{mss423@umd.edu}

\begin{abstract}
In this paper, we analyze a natural learning algorithm for uniform pacing of advertising budgets, equipped to adapt to varying ad sale platform conditions. On the demand side, advertisers face a fundamental technical challenge in automating bidding in a way that spreads their allotted budget across a given campaign subject to hidden, and potentially dynamic, cost functions. This automation and calculation must be done in runtime, implying a necessarily low computational cost for the high frequency auction rate. Advertisers are additionally expected to exhaust nearly all of their sub-interval (by the hour or minute) budgets to maintain budgeting quotas in the long run. To resolve this challenge, our study analyzes a simple learning algorithm that adapts to the latent cost function of the market and learns the optimal \textit{average} bidding value for a period of auctions in a small fraction of the total campaign time, allowing for smooth budget pacing in real-time. We prove our algorithm is robust to changes in the auction mechanism, and exhibits a fast convergence to a stable average bidding strategy. The algorithm not only guarantees that budgets are nearly spent in their entirety, but also smoothly paces bidding to prevent early exit from the campaign and a loss of the opportunity to bid on potentially lucrative impressions later in the period. 

In addition to the theoretical guarantees, we validate our algorithm with experimental results from open source data on real advertising campaigns to further demonstrate the effectiveness of our proposed approach.
\end{abstract}


\newcommand{\BibTeX}{\rm B\kern-.05em{\sc i\kern-.025em b}\kern-.08em\TeX}

\begin{document}


\pagestyle{fancy}
\fancyhead{}


\maketitle 


\section{Introduction}
\label{sec.intro}
Ad impressions sold through real-time bidding (RTB) auctions are responsible for an ever-increasing portion of company expenditures as well as the revenue of large ad providers like Google, Facebook, Amazon, Microsoft, and Yahoo!. Google alone is responsible for upwards of 50 billion ad impressions on average per day \cite{Tunuguntla_2021}, with a corresponding revenue on the order of \$100 million. While companies bidding in these advertising campaigns do not participate in every auction throughout the day, they are often required to participate enough to spend their allotted budget in this time. Thus, great interest is placed on adequately pacing budgets throughout the day so as to not spend too hastily at the beginning of a day and miss out on better impressions, or spend too frugally until the close of day. 

We here take the perspective of a demand-side platform (DSP), an intermediary serving as the interface connecting advertisers with
ad-exchanges and offering customized bid strategies in online auctions. Our model of a DSP is given a users budget and target spent amount, and is asked to uniformly pace spending so as to exhaust this amount.


In optimizing bids within an RTB, an advertiser hopes to maximize  key performance indicators, such as cost per conversion (CPC) or click through rate (CTR), for an ad campaign subject to the budgeting constraints for each time period \cite{Cai_2017,Jauvion_2020,lobos2018optimal,Zhang_2014}. By smoothly spreading a budget across the duration of a campaign, an advertiser can ensure sustained influence across the broad spectrum of ad impressions with which they interact.
We can summarize the objective of an advertiser using a DSP as follows:
\begin{itemize}[leftmargin=*]
    \item \textbf{Achieve performance goals:} in an advertising campaign, performance goals are typically meant to ensure a broad audience is reached while maintaining a low CPC, ie. the investment required to acquire new customers is relatively small.
    \item \textbf{Optimally pace the budget:} an advertiser hopes to have a sustained influence in an ad campaign, smoothly rolling out their advertisements to ensure impact on users. This budget pacing can be uniform (equal representation throughout a campaign) or potentially more sophisticated so as to maximize the number of users reached during certain times of the day, while placing less priority on other periods.
\end{itemize}

In spite of the ever increasing prevalence of online advertising in modern society and the clear value of optimizing the above objectives in practice, the research literature on the problem instance of optimally spending budgets in their entirety is limited. The bulk of the work either excludes comprehensive analysis of the stability of a proper budget pacing mechanism in favor of simulating these results \cite{Cai_2017,lee2013,Xu_2015}, or does not impose the realistic constraints that are met in the advertising economy \cite{tapia_2015,lobos2018optimal}. As such, we bring an analytic approach in combination with simulated results on real datasets of our pacing mechanism to bridge this gap in the literature.

\subsection{Main Ideas}
\label{sec.intro_main}

In this paper, we present an online approach to smoothly and uniformly pacing an advertiser's budget over a fixed-length advertising campaign. Our approach invokes an iterative control feedback mechanism to estimate the proper \emph{average} bid to submit over each time period, which is later manipulated by a mechanism (potentially hidden to the bidder) to compute an ``actual spent amount" for that period so that the total budget is approximately spent in a uniform fashion throughout the campaign. The algorithm relies on simplistic scaling of bids in response to learning of this latent mechanism in a naturalistic way, and is currently in implementation at 
two major companies with which the authors were previously affiliated. Specifically, we analytically derive the stability and convergence time of this natural learning algorithm that has analogously appeared in other online bidding literature \cite{lee2013,Zhang_2014}. Though simple, this algorithm exhibits low convergence times for the problem of budget pacing and possesses interesting dynamics in accordance with other one-dimensional mappings like the logistic map. We summarize our main results informally:

\begin{itemize}[leftmargin=*]
    \item Develop an iterative control feedback algorithm for smart pacing of an advertisers budget through an ad campaign.
    \item Analytically derive the efficiency of this algorithm in quickly converging to an optimal bid strategy to ensure budgeting quotas are met.
    \item Implement our algorithm with extensive simulations on real ad campaigns from Yahoo! to demonstrate its feasibility in practice.
\end{itemize}


The remainder of the paper is organized as follows. In Section \ref{sec.back}, we discuss the budgeting problem on a high level and previous work in the field. In Section \ref{sec.opt}, we describe our model and give a comprehensive analysis of the convergence dynamics exhibited. We discuss extensions of our algorithm to account for other realistic budgeting schemes to be implemented by a DSP in Section \ref{sec.realworld}, and finally present simulations of our algorithm on real datasets in Section \ref{sec.simul} to demonstrate its utility in practice. We conclude with a discussion of our methodology, its limitations, and where future work will be directed.

\eat{
In Section \ref{sec.simul}, we compare the theoretical insights with simulated results and further discuss the extensions of our implementation to capture different pacing schemes in Section \ref{sec.realworld}. Lastly, in Section \ref{sec.conc}, we conclude with a discussion of our methodology, its limitations, and where further work is needed.
}

\section{Background and Related Work}
\label{sec.back}

In this section, we discuss the problem of budget pacing with the joint problem of bid optimization, and subsequently discuss the related work, noting where these results fall short of the reality of RTB schemes. 

\subsection{Problem Statement}
\label{sec.back_prob}
We consider \emph{online bid optimization} in the following framework: there are $T$ auction periods ordered by an index $t \in \{1, ..., T\}$. At the start of an auction period, an individual advertiser must make a decision as to how much to bid for desired ads within this period, denoted by $b_t$. Advertisers further have a total daily budget $B \in \mathbb{R}$ that limits the amount that can be spent within the day. Typically, advertisers would like to have smooth budget delivery \cite{Agarwal_2014,Avadhanula_2021,Balseiro_2019,Balseiro_2021,Conitzer_2021,lee2013,Xu_2015}, expressed as not buying more than a set fraction of the impressions before a given time in order to ensure that:
\begin{enumerate}
    \item budgets are not prematurely expended, thus resulting in missed opportunities later in the day
    \item spending does not fluctuate substantially for ease of analysis.
\end{enumerate}

As a simple example, assume that each $t$ denotes a 2 minute ad period on a streaming website. Within this time frame, advertiser $A$ may choose to bid with value $b_t$ in auctions against all advertisers interested in showing an ad to a given subset of users. Within a split-second, the highest bidder for each user is chosen and $A$ must pay for the auctions they won, incurring some cost $$c_t = b_t \times \text{ (number of auctions won)}.$$ 
Following observation of this incurred cost, the advertiser $A$ invokes an algorithm to adjust their bidding in an effort to ensure a smooth depletion of their budget throughout the day.

\eat{
We consider an online bid optimization problem subject to budgeting constraints in the following way: throughout a given campaign, there are $T$ total auction periods being conducted sequentially (indexed by $t \in \{0,2,...,T-1\}$). For example, these periods could correspond to each minute in a day, which would amount to $\approx 10,000$ auctions in each period. An advertiser must decide how much to bid on average for impressions during each period of the campaign, denoting their average bid as $b_t$ for the $t$-th period.
}

We note that for each impression, the base bid $b_t$ can be rescaled in a manner corresponding to the value of the impression for that advertiser. For example: as formulated in \cite{Perlich_2012}, the bid value of the $i$-th impression is rescaled by the probability of its conversion $p_i$, divided by the average conversion (denoted $\overline p$). This is often done for individual impressions of great interest to an advertiser, however, in our formulation, $b_t$ is submitted \textit{on average} in each  auction period. 

\eat{
We assume an advertiser has a budget $B > 0$ for the $T$ auctions. The key aspect of real-time bidding in a finite time period is that each advertiser is trying to both deliver their budget smoothly and spend it nearly entirely \cite{Agarwal_2014,Avadhanula_2021,Balseiro_2019,Balseiro_2021,Conitzer_2021,lee2013,Xu_2015}.
}

The added complication of the smooth delivery problem we focus on in the present study is that oftentimes, when an advertiser submits a bid in an auction, this may not be the true value they pay for that impression. As demonstrated in the simple example above, although $A$ submits a bid $b_t$, they actually incur a cost that is potentially much larger (proportional to the number of auctions won).

\eat{
To reiterate this novel context, consider the simplistic instance where the advertising campaign is conducted using a first price auction. The amount spent in each time period is thus the proportion of auctions won multiplied by the submitted bid and thus the underlying actual spent amount function is a linear scaling of the submitted bids and generalizes the expenditure across hundreds of thousands of auctions. 
}
In response, the pacing of an advertiser's budget relies on learning the actual cost of submitted bids to adequately scale bid values to meet desired spending goals in each period. 

We can now formally frame the budget pacing problems as:
\begin{align}
    \textbf{minimize } &B - \sum_{t=1}^T c_t \\
    \textbf{subject to } &\left|\frac{B}{T} - c_t\right| \leq \epsilon  
\end{align}
for some $\epsilon \ll 1$ and where the optimization is framed using the information for all the auctions, however, the problem itself is online. As a result, it is clear we need an algorithm that quickly learns the latent cost function, $f(b_t) = c_t$, in a small portion of the total number of auctions and subsequently uniformly paces their bidding for the duration of the campaign. Of crucial importance, we note that the cost function in each period implicitly contains information on the number of impressions won and overall surplus generated based on the value of impressions won, subject to the payment rules of whichever auction mechanism is implemented for those impressions. 

Lastly, it is important to note that the cost function, $f$, is \emph{subject to change} at different times in a campaign and can be adjusted by the advertiser when manipulating their target spend amount for different times. For example, if the algorithm is being utilized to exhaust a daily budget, the advertiser may emphasize impressions in the morning and input a higher target spend amount for that portion of the day, while reducing this target for later times in the day. With slight modifications of our algorithm we can further capture these instances (see Section \ref{sec.realworld}).

\subsection{Related Work}
\label{sec.back_rel}

Real-time bidding strategies have an expansive literature \cite{Celli_2021,Feng_2018,Geyik_2015,Karaca_2019,Nedelec_2019,Weed_2016}, however the field of optimal budget pacing is relatively new and, as such, the literature thus far is often limited in scope \cite{Agarwal_2014,Balseiro_2019,Balseiro_2021,borgs_2007}. We here examine the methodologies of a select number of other studies whose work closely aligns with our own, the aspects of real-time bidding they encompass, and where they fall short. 

Most closely related to our work is that of Lee et al. \cite{lee2013}, which devises a comparable algorithm for pacing the budget while also maximizing the performance indicators of impressions. In particular, their algorithm scales bids between successive rounds to ensure smooth delivery of the budget without early exit. However, this paper maximizes impression value by estimating a threshold for advertisers to begin bidding. As a result, the algorithm bids more highly on the valuable impressions, but risks overspending and early exit from the campaign in these instances. Our algorithm instead ensures that agents are present throughout the entire campaign so that no impression is missed out on and guarantees sustained presence in an ad space. Additionally, the work of Lee et al. does not account for the real-world quotas that advertisers must meet, forcing them to spend their budgets in close to their entirety. Our algorithm accounts for this constraint and utilizes it in the scaling of bid values. 

The work done by Xu et al. \cite{Xu_2015} also aligns well with our study, emphasizing smooth delivery of budget and minimizing the probability of early exit from a campaign to ensure a sustained presence for advertisers. When omitting performance goals, their algorithm is comparable to ours as it optimizes bidding by scaling bids based on the discrepancy between desired spending amount and the actuality. Additionally, they optimize over a fixed penalty function, similar to our latent cost function. However, their algorithm is not robust to changes in this penalty over time and fails to provide analytic results on the convergence and stability of this algorithm. Our study attempts to bridge the gap between the analytic and computational studies on optimal budget pacing. Lastly, this paper fails to consider the real-world constraint of budgeting quotas. 

Fernandez-Tapia \cite{tapia_2015} examines the pacing problem from  different perspective to the above, deriving analytical results through variational calculus techniques. This work compares the optimal bidding strategy when RTBs occur at a linear rate (fixed spacing between bids) and nonlinear. As most papers, including our own, make the simplifying assumption that auctions occur at such a high rate that we need not consider unequal arrival of the impressions to bid on, this more technical study highlights a more complex and realistic setting of auctions. However, while impressions may arrive in this more complicated fashion, advertisers do not in practice base their strategies around this assumption. As such, our algorithm is better suited for real-world implementation.

The problem of finding an optimal bidding strategy for an advertiser who does not know their own valuation, subject to budgeting constraints, has also been studied as an extension of the multi-armed bandit problem as in \cite{Avadhanula_2021}. Avadhanula et al. provide bidding schemes over discrete and continuous bid spaces for $m$ platforms simultaneously with regret lower bounds of $\Omega \left( \sqrt{mOPT} \right)$ and $\Omega \left( m^{1/3}B^{2/3} \right)$ -- where $B$ is the budget and $OPT$ the performance of the optimal bidding strategy. This model however is limited in that the auction format is restricted to a second-price auction payout mechanism, whereas our approach provides a more general framework which can be applied to any auction mechanism. 

\section{Optimal Bidding}
\label{sec.opt}
In this section, we detail the algorithm for learning the latent cost function and uniformly pacing an advertiser's budget throughout an advertising campaign. We first present the algorithm with the intuition behind its definition and proceed by analyzing the convergent dynamics and efficiency. We emphasize uniformly pacing an advertiser's budget, or target spend amount, so that all auction periods are weighted equally. The primary concern is to spend the entirety of their budget while staying active in the campaign for as long as possible. 

We define the metrics by which a budget pacing algorithm should be measured in order to assess its overall effectiveness in online real-time bidding systems. 
\begin{enumerate}
	\item Fast convergence to a stable bidding strategy relative to the overall campaign duration.
	\item Adaptive capacity for changes in the display auction mechanism (manipulated on the supplier end).
\end{enumerate}

Ensuring good performance metrics is complex in the online market due to the speed and frequency that impressions are sold, yet we prove in this section that our natural learning algorithm performs well due in large part to its simplicity and efficient computation.

\subsection{Pacing Algorithm}
\label{sec.opt_algo}

We begin by assuming an advertiser has two pieces of information: its budget, $B$, and the number of auction periods in the campaign, $T$ (ie. auction periods per day). Intuitively, the advertiser who is trying to uniformly pace their budget will initially bid its average budget for the corresponding number of auctions, $b_0 = \frac{B}{nT}$ (where $n$ is the number of auctions per period). However, in general, the actual cost is much larger than the input bid and the convergence time remains low regardless of this initial bid selection, so any choice suffices. Once an initial average bid is set, it is utilized for the first period ($t=0$), once again being scaled for each impression in accordance with their value, and the agent incurs cost $c_0$ throughout this time. 

Following the initial bid, an advertiser now has the information of how much was spent in the first period and can assess the discrepancy between the \emph{desired} amount to be spent and the \emph{actual} amount. Let $B_r^t$ denote the remaining budget after period $t$, and $c_{t}$ be the incurred cost in this period, then we can define the scaling factor $\frac{B_r^t}{T-t} \cdot \frac{1}{c_t}$ as the ratio of the amount the advertiser wants to spend on average in each remaining auction period to how much it spent on the previous. Using this ratio as a budget pacing factor, we have the following iterative scheme:
\begin{align*}
    b_{t+1} = \frac{B_r^t}{c_t (T-t)} \cdot b_t
\end{align*}

\begin{algorithm}[t]
\caption{Budget Smoothing}
\label{alg:algorithm}
\begin{algorithmic}[1]
\STATE \textbf{Input}: $B, T, t, b_t, (c_0, ..., c_t)$ \\
\STATE \textbf{Output}: $b_{t+1}$
\STATE $B_r^t = B - \sum_{i=0}^t c_i$ \COMMENT{remaining budget after time $t$}
\STATE $c_{\text{opt}} = \frac{B_r^t}{T-t}$ \COMMENT{optimal spend amount}
\STATE $c_{\text{act}} = c_t$ \COMMENT{actual spend amount}
\STATE $\alpha = \frac{c_{\text{opt}}}{c_{\text{act}}}$
\STATE \textbf{return} $\alpha \cdot b_t$
\end{algorithmic}
\end{algorithm}

\textsc{Algorithm} \ref{alg:algorithm} takes as input the budget ($B$), number of periods in a campaign ($T$), current period ($t$), the most recent bid value ($b_t$), and the history of previously spent amounts ($\{c_i\}_{i=0}^t$), outputting a new average bid value scaled based on this information. Note that each $c_i$ is equal to $f(b_i)$, however the algorithm is trying to \emph{learn} $f$, so we use the former notation.

As we will see in Section \ref{sec.opt}, this intuitive update algorithm converges to a fixed point that uniformly paces an advertisers budget within a small fraction of the total auction time, as well as undergoing interesting bifurcations when parameters of the system are tuned. A brief summary of the results for latent cost functions in the family $f(b_t) = C \cdot b_t^k$, where $k \in \mathbb{R}$ and $C \in \mathbb{R}_{+}$, are provided in Table \ref{tab:1}. 

\begin{table*}[t]
\center
\begin{tabular}{c|c}
\textbf{Parameter Range} & \textbf{Dynamics}                                     \\ \hline
$k \leq 0$               & Unstable fixed point                                  \\ \hline
$k \in (0,1)$ & Stable fixed point with convergence time $\propto \frac{\ln(k)}{\ln(1-k)}$   \\ \hline
$k = 1$                  & Stable fixed point with convergence in one iterations \\ \hline
$k \in (1,2)$ & Stable fixed point with convergence time $\propto \frac{\ln(2-k)}{\ln(k-1)}$ \\ \hline
$k \geq 2$               & Instability requiring guard rails on spending        
\end{tabular}
\caption{Summary of convergence results for cost functions of the form $f(b_t) = C \cdot b_t^k$.}
\label{tab:1}
\end{table*}

In analyzing the efficiency of this learning algorithm, we consider various reasonable actual cost functions $f(b_t) = c_t$ and their divergent or convergent behavior to an optimal strategy. In the following analysis we consider only \textit{continuous} cost functions, and as such simplify our analysis by examining polynomial functions which can approximate any such function -- a consequence of the Stone-Weierstrass theorem \cite{rudin}.

\subsection{Linear and Semi-Linear Cost Function}
\label{sec.opt_lin}
As a warm-up, we first consider the most simplistic form of the latent cost function: a \textit{linear} cost. We assume that $f(b_t) = C \cdot b_t$ for some $C > 0$: the situation in which the agent running the ad sale simply scales incoming bids by a fixed constant factor. We reiterate that such a simple scaling can effectively capture the dynamics of a repeated first auction system wherein the actual amount spend by a bidder is proportional to the number of auctions won in a given period.

In this elementary case, we see that, for any initial bid placed an advertiser adapts to the latent cost function in \textit{exactly} one iteration, bidding the optimal uniform amount for the remainder of the campaign.
\begin{theorem} \label{thm:linear}
For a linear cost function, $f(b_t) = C \cdot b_t$ where $C > 0$, and initial bid $b_0$, \textsc{Algorithm} \ref{alg:algorithm} converges to a fixed point bid value in exactly one iteration.
\end{theorem}

In many practical instances, it is also important to consider a \textit{semi-linear} cost function. We define the semi-linear cost function as $f(b_t) = \min\{C \cdot b_t,M\}$\footnote{Additionally, if we use the semi-linear function with a maximum rather than minimum, we see identical convergence rates.} where both $C,M > 0$. This is analogous to when the administration of the ad sales both scale an advertisers bids by a fixed constant factor and also instate "guard rails" on the spending: requiring the amount spent to be no greater than $M$. This is meant to protect against over excessive spending by an advertiser. Despite this limitation on bids, we see that the semi-linear function converges in one iteration when $M > \frac{B}{cT}$.

\begin{theorem} \label{thm:semi_lin}
For a semi-linear cost function, $f(b_t) = \min\{C \cdot b,M\}$  where $C > 0$, initial bid $b_0 < M $, and $M  > \frac{B}{CT}$ \textsc{Algorithm} \ref{alg:algorithm} converges to a fixed point average bid value in exactly one iteration.
\end{theorem}
Additionally, for $M \gg 1$, the semi-linear function is identical to the linear case. As a result, proving Theorem \ref{thm:semi_lin} encompasses the result of Theorem \ref{thm:linear}.
\begin{proof}[Theorem \ref{thm:semi_lin}]
	Assume $b_0 \neq b^*$:
	\begin{align*}
	b_1 &= \text{\textsc{Algorithm} 1}(B,T,1,b_0,c_0) \\
	&= \frac{B - \min\{Cb_0,M\}}{\min\{Cb_0,M\}T} \cdot b_0 \\
	&= \frac{B - Cb_0}{cT} \\
	b_2 &= \text{\textsc{Algorithm} 1}(B,T,2,b_1,(c_0,c_1)) \\
    &= \frac{B - \min\{Cb_0,M\} - \min\{cb_1,M\}}{\min\{cb_1,M\}(T-1)} \cdot b_1 \\
    &= \frac{B - Cb_0 - \min\{C \cdot \frac{B - Cb_0}{CT},M\}}{\min\{C \cdot \frac{B - Cb_0}{cT},M\}(T-1)} \cdot \frac{B - Cb_0}{CT} \\
    &= \frac{B - Cb_0}{CT} \\
    &= b_1
\end{align*}
\end{proof}

\begin{proof}[Theorem \ref{thm:linear}]
	Using the result of Theorem 3.2 we take the limit as $M \rightarrow \infty$, thus giving $\min\{cb,M\} = cb$ with convergence in one iteration. 
\end{proof}

Before proceeding to give theoretical bounds on the convergence times of our algorithm for the cases where we do not see one iteration convergence, we must first define the ``general cost function" and assess the fixed points of our system.

\subsection{General Cost Function}
\label{sec.opt_gen}

In the most complex case, we consider a general cost function of the form $f(b_t) = C \cdot b_t^k$ where $k \neq 1$. In this case, we here identify the fixed point of the system and assess its stability. 

\begin{theorem} \label{thm:stable}
For $k \in (0,2)$, Algorithm 1 has exactly one stable fixed point for cost functions of the form $f(b_t) = C \cdot b_t^k$ where $C > 0$.
\end{theorem}
\begin{proof}
The fixed point of our system occurs when the output of Algorithm 1 is equal to the input $b_t$. Solving this equation yields:
\begin{gather*}
    b^* = \left(\frac{B - \sum_{i=0}^{t-1}Cb_i^k}{c(T-t+1)}\right)^{1/k} \tag{3.3.1} \label{eq:fpt}
\end{gather*}

Noting that in general, if $b_0$ is not a fixed point, then the fixed point for period $t$ is defined as the above. If instead $b_0$ is a fixed point, the system will thus bid this value on average for the entire campaign and the point has the more simplistic formula derived by the following:
\begin{gather*}
    b^* = \left(\frac{B}{C(T+1)}\right)^{1/k} \tag{3.3.2} \label{eq:fp1}
\end{gather*}

which is the special case of \eqref{eq:fpt} where $t=0$. As such it suffices to examine the more general formula to encompass all results. 

In order to prove stability of the above fixed points, we linearize around the point $b^*$ and assess whether the point is attracting or repelling by examining the derivative at that point \cite{strogatz:1994}. Informally, if we let $|f'(b^*)| = |\lambda|$, then we have linear stability when $|\lambda| < 1$ and instability when $|\lambda| > 1$. Now we can analyze our system's fixed point. First we differentiate the function in \textsc{Algorithm} \ref{alg:algorithm} to get the following:

\begin{gather*}
    \frac{d}{db_t}\left(\frac{B-\sum_{i=1}^t Cb_i^k}{Cb_t^k(T-t)} \cdot b_t\right) \\
    = \frac{-C + (1-k)b_t^{-k}(B-\sum_{i=1}^{t-1}Cb_i^k)}{C(T-t)} \tag{3.3.3} \label{eq:deriv}
\end{gather*}

Plugging in the fixed point from \eqref{eq:fpt} gives the following stability estimate respectively

\begin{align*}
    |\lambda| = \left|1 - k\left(\frac{T-t+1}{T-t}\right)\right| \approx |1-k| \text{ for } t \ll T \tag{3.3.4} \label{eq:stab}
\end{align*}

Thus for $0 < k < 2$, the system is in general attracted to the fixed point defined in \eqref{eq:fpt}, and the point is unstable outside of this parameter range.
\end{proof}

\subsection{Convergence Time Bound}
\label{sec.opt_lb}

The crucial result of this paper comes from the minimal convergence time for \textsc{Algorithm} \ref{alg:algorithm}. In contrast to recent papers examining optimal online budget pacing schemes, none have analytically derived an estimate for convergence times to these optimal strategies \cite{lee2013}.

\begin{theorem} \label{thm:conv}
For $|1 - k| < 1$, \textsc{Algorithm} \ref{alg:algorithm} has a bounded distance from the stable bid value at time $t$ defined by:
\begin{align*}
    \epsilon := |b_t - b^*| \leq \gamma^{-1/k} \cdot \frac{|1-k|^{t - 1 + \frac{1}{k}}}{1-|1-k|}
\end{align*}
where $\gamma = \frac{CT}{B}$. Subsequently, we have a convergence time, $t^*$, to the stable bid value bounded as:
\begin{align*}
	t^* \leq \frac{k-1}{k} + \frac{\ln\left|\epsilon \gamma^{1/k} (1 - |1-k|)\right|}{\ln|1-k|} \tag{3.4.1} \label{eq:conv}
\end{align*}
\end{theorem}
\noindent We can see that the convergence time is dependent upon the parameters of our RTB campaign, namely the budget, length, and degree of the latent cost function. In order to prove the theorem, we first need a crucial lemma from the theory of metric spaces, proven in \cite{fixpt_thrm}.

\begin{lemma}[Banach Fixed Point Theorem] \label{lemma:banach}
Assume that $\Omega \subset \mathbb{R}^n$ is closed, and that $G: \Omega \mapsto \Omega$ is a contraction, that is, there exists $0 \leq L < 1$ such that
\begin{align*}
    \|G(x) - G(y)\| \leq L\|x - y\|,
\end{align*}
for all $x,y \in \Omega$, where $\|\cdot\|$ is any norm on $\mathbb{R}^n$. Then the function $G$ has a unique fixed point $x^* \in \Omega$. Additionally, let $x_0 \in \Omega$ be arbitrary and define the sequence $\{x_t\}_{t=1}^{\infty} \subset \Omega$ by $x_t = G(x_{t-1})$. Then we have the estimate
\begin{align*}
    \|x_{t} - x^*\| &\leq \frac{L^t}{1-L}\|x_1 - x_0\|
\end{align*}
In particular the sequence $\{x_t\}_{t=1}^{\infty}$ converges to $x^*$
\end{lemma}

Invoking this well studied result from Banach requires only that the mapping is contractive. We henceforth proceed in proving the theorem by demonstrating this crucial property of \textsc{Algorithm} 1, and applying Lemma \ref{lemma:banach}.
\begin{proof}[Theorem \ref{thm:conv}]
To invoke Lemma \ref{lemma:banach} we need to estimate an adequate Lipschitz constant $L$ and put a proper bound on the distance between two points in the sequence defined by our algorithm. We first note that 
\begin{gather*}
    |G(x) - G(y)| \leq L|x - y| \\
    \Rightarrow \frac{|G(x) - G(y)|}{|x - y|} \leq L \leq |G'(x^*)|
\end{gather*}
where $x^* \in (x,y)$ gives us an easily computable bound the sequence when $|G'(x)| < 1$. Note that we now use the absolute value rather than other norms since our algorithm uses scalar valued functions. 

Stemming from our analysis in Section \ref{sec.opt_gen}, a Lipschitz bound for the dynamics when $0 < k < 2$ is $L = |1-k|$. The remaining component needed to apply Lemma \ref{lemma:banach} is a bound on the distance $|b_1 - b_0|$. For any $b_0$ that is not the fixed point, we can look for the maximal distance to $b_0$ by taking the derivative of $b_1 - b_0$ and setting it equal to 0. Assume without loss of generality that $b_1 - b_0 > 0$.

\begin{align*}
    \frac{d}{db_0}|b_1 - b_0| &= \frac{d}{db_0}\left|\frac{B-Cb_0^k}{Cb_0^k(T-1)} \cdot b_0 - b_0\right| \\
    &= \frac{-C + |1-k|b_0^{-k}B}{C(T-1)} - 1
\end{align*}
Setting the derivative equal to 0 and finding the maximal $b_0$
\begin{align*}
	b_0 = \left(\frac{B|1-k|}{CT}\right)^{1/k}
\end{align*}
which gives the distance value
\begin{align*}
    |b_1 - b_0| = \left|\frac{T-|1-k|}{|1-k| \cdot (T-1)} \cdot \left(\frac{B|1-k|}{CT}\right)^{1/k}\right| \tag{3.4.2} \label{eq:maxdis}
\end{align*}

Now plugging in our Lipschitz constant and maximal distance from \eqref{eq:maxdis} to the estimate from Lemma \ref{lemma:banach} yields
\begin{align*}
	|b_t - b^*| &\leq \frac{|1-k|^{t-1}}{1-|1-k|} \cdot \left|\frac{T-|1-k|}{T-1} \cdot \left(\frac{B|1-k|}{CT}\right)^{1/k}\right| \\
	&\approx \frac{|1-k|^{t-1}}{1-|1-k|} \cdot \left|\left(\frac{B|1-k|}{CT}\right)^{1/k}\right|
\end{align*}
Lastly, we can define $\gamma := \frac{CT}{B}$ to get the simplified
\begin{align*}
	|b_t - b^*| \leq \gamma^{-1/k} \cdot \frac{|1-k|^{t - 1 + \frac{1}{k}}}{1-|1-k|}
\end{align*}
and finally, we can rearrange this inequality to give us a bound on the convergence time:
\begin{align*}
	t &\leq \frac{k-1}{k} + \frac{\ln\left|\gamma^{-1/k} (1 - |1-k|) |b_t-b^*| \right|}{\ln|1-k|} \\
	&= \frac{k-1}{k} + \frac{\ln\left|\epsilon \gamma^{1/k} (1 - |1-k|)\right|}{\ln|1-k|}
\end{align*}
thus, we have the result.
\end{proof}

Lastly, it is important to note that up to this point, we have only considered latent cost functions which are monomial (consisting of one term). This is a simplifying assumption which is justified by bounding polynomial cost functions with monomials, detailed in the following lemma (the proof can be found in Appendix \ref{sec:appendix}).
\begin{lemma} \label{lemma:bounded}
For a cost function defined using the general form $c_1b^{k_1} + c_2b^{k_2} + ... c_mb^{k_m}$ where $k_1 > k_2 > ... > k_m$ and $b > 0$, the convergence time of \textsc{Algorithm} \ref{alg:algorithm} is bounded by that of monomial functions.
\end{lemma}

\subsection{Loss of Stability} 
\label{sec.opt_insta}
As noted in Section \ref{sec.opt_gen}, when $k > 2$ the fixed point for average bidding becomes unstable. As a result, \textsc{Algorithm} \ref{alg:algorithm} does not converge for a cost function of quadratic or higher degree. In this instance however, we can once again utilize guard rails on spending to obtain convergence to cyclic bidding behavior. The following theorem, proven in Appendix \ref{sec:appendix}, demonstrates the importance of guard rails on spending to ensure adequate budget pacing for advertisers adapting to the RTB auction mechanism. The proof invokes the same strategy as Theorem \ref{thm:linear}, where we now find $b_0$ which are fixed points of \textsc{Algorithm} \ref{alg:algorithm} applied twice, and can be found in the appendix.
\begin{theorem} \label{thm:chaos}
For a cost function of the form min$\{C \cdot b^k,M\}$ where $C >0$ and $\frac{B}{CT} < M < B$, \textsc{Algorithm} \ref{alg:algorithm} undergoes a period doubling bifurcation at $k = 2$, leading to oscillatory behavior between two average bid values for each iteration $t$.
\end{theorem}
This period doubling is further indicative of a canonical path to chaos \cite{strogatz:1994} and as such a loss of stability for our system.


\section{Real-World Implementation}
\label{sec.realworld}
The algorithm presented is easy to implement and can be specifically tailored to a wide range of applications. As noted previously, slight variants on the presented algorithm are used at several companies which tailor the bidding to tertiary interests that we discuss in this section. \footnote{Variations of the bid pacing algorithm presented in this paper are currently implemented at two major companies with which an author was previously involved. Due to proprietary data restrictions at these companies, we focus mainly on important theoretical aspects of the algorithm along with simulations results in this paper.} We focus on two specific variations: fluctuating target spend amounts and subthreshold budgets, however the simplicity of our algorithm allows for a plethora of adjustments for specific needs.

\subsection{Non-Uniform Pacing}
While our main analysis demonstrates the stability of average bidding for the desired goal of uniform pacing, the algorithm can be simply adapted to meet changing target spend amounts throughout the campaign. For instance, if an advertiser wants to decrease spending during the morning since the bulk of their target market is not online, they may decrease their target spend amount until later in the day. A variant of this form is easily achieved by adjusting the submitted budget for any given time period(s). This is achieved by adding another input parameter $\kappa_t \in (0,1]$ which scales the budget by the multiplier at any given time point. Pseudocode for such a variation on our algorithm is presented in Algorithm \ref{alg:algorithm_nonu}. We note that in spite of this tailoring to a specific need, the convergence properties that were theoretically justified remain, and will be further validated in Section \ref{sec.simul}.

\begin{algorithm}[t]
\caption{Non-Uniform Pacing}
\label{alg:algorithm_nonu}
\begin{algorithmic}[1]
\STATE \textbf{Input}: $B, T, t, b_t, (c_0, ..., c_t), \kappa_t$ \\
\STATE{$B_v \leftarrow \kappa_t \cdot B$} \\
\STATE{$b_{t+1} = $ \textsc{Algorithm 1($B_v, T, t, b_t, (c_0, ..., c_t)$)}}
\STATE \textbf{Output}: $b_{t+1}$
\end{algorithmic}
\end{algorithm}

\begin{algorithm}[t]
\caption{Subthreshold Budget}
\label{alg:algorithm_sub}
\begin{algorithmic}[1]
\STATE \textbf{Input}: $B, T, b_t, (s_0, ..., s_t)$ \\
\IF{$B < \tau$ for some threshold $\tau$}
	\STATE{$B_v \leftarrow \sigma \cdot B$ for some $\sigma > 1$ fixed by the provider} \\
	\STATE{$b_{t+1} =$ \textsc{Algorithm 1($B_v, T, b_t, (s_0, ..., s_t)$)}}
\ELSE
	\STATE{$b_{t+1} =$ \textsc{Algorithm 1($B, T, b_t, (s_0, ..., s_t)$)}}
\ENDIF
\STATE \textbf{Output}: $b_{t+1}$
\end{algorithmic}
\end{algorithm}

\subsection{Subthreshold Budgets}
In the instance where users submit a budget to the DSP that are too small for meaningful uniform pacing, the provider can implement a variant on our algorithm for forcibly exhausting a budget early in the campaign period, thus resulting in the user needing to update with a larger budget for the next campaign. \footnote{We once again note that this is common practice at certain companies, in an effort to force users to meet a minimum budget threshold.} This is achieved by implementing a ``virtual budget" where the DSP simply scales up the input budget and thus the algorithm establishes a higher average bid than is feasible for the user's actual budget, leading to early exit from the campaign. While this outcome is counter to the presented practicality of the algorithm, it is an effective means by which a provider can ensure that advertisers are submitting a high enough budget to be competitive in the market space. This variant's pseudocode is presented in Algorithm \ref{alg:algorithm_sub}. We reiterate that although the \emph{intent} of this adapted algorithm is not in line with the primary justification for our presented results, it nonetheless retains the convergence properties of Section \ref{sec.opt} and will be validated through experimentation in Section \ref{sec.simul}.

\section{Simulation Results}
\label{sec.simul}
In this section, we proceed to validate our theoretical results through comprehensive simulations on both real and simulated data. We demonstrate the discussed convergence dynamics of the bidding scheme as well as implement our DSP within the context of a Yahoo! RTB auction system from publicly available data. All codes for the presented results are available as supplementary material to our paper.

\subsection{Bid Dynamics}
We here validate our model against simulated results for various different parameter values. These experiments are meant to illustrate the empirical behavior of our simplified model. We present the dynamics of bid values as a function of time in Figure \ref{fig:bid_dynamics} for the three algorithms discussed above. Namely, the first column plots the dynamics of Algorithm \ref{alg:algorithm}, with columns two and three depicting Algorithms \ref{alg:algorithm_nonu} and \ref{alg:algorithm_sub} respectively. 

For the presented simulations, the advertisers is given an arbitrary initial average bid for the first period, and subsequently bids based on the corresponding algorithms. We henceforth set the budget for all simulations at $B = 50000$, $T = 1000$, and an amount spent function defined as $f(b) = \min\{b_t^k,100\}$. Additionally, we once again make the simplifying assumption that the spent functions are monomials, allowing us to study the dynamics exclusively as a function of the degree $k$.

\begin{figure*}[h]
	\center{\includegraphics[width=0.68\linewidth]{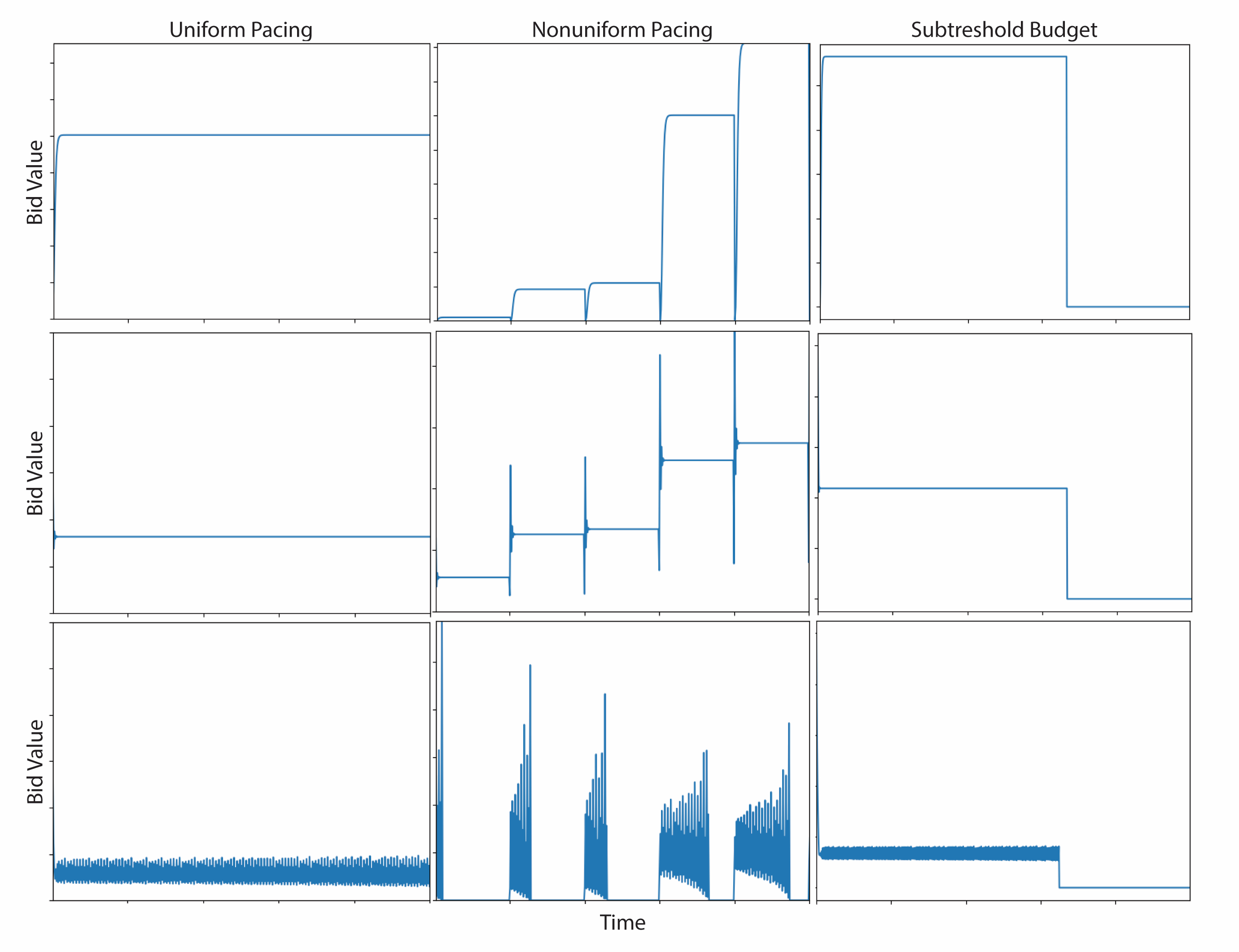}}
    \caption{\label{fig:bid} Bid values for different $k$ and algorithms. The first row plots the bid dynamics of Algorithms \ref{alg:algorithm}, \ref{alg:algorithm_nonu}, \ref{alg:algorithm_sub} for $k = 0.5$. The second and third rows plot the dynamics for $k = 1.4$ and $2.3$ respectively.}
    \label{fig:bid_dynamics}
\end{figure*}

We first consider $k = 0.5$, evaluating \eqref{eq:conv} with the given parameter set gives a bound on the iterations to convergence of $t^* \leq 31$. In simulating the model (row 1 of Figure \ref{fig:bid_dynamics}) we see convergence within an error tolerance $\epsilon = 10^{-6}$ in 23 iterations for Algorithm \ref{alg:algorithm}, approximately 3\% of the total auction period with 99.90\% of the budget spent. Additionally, in the non-uniform instance where we force the system to increase bidding as the day progresses via Algorithm \ref{alg:algorithm_nonu}, we see that after each pacing increase the algorithm quickly converges to a new stable state until the pacing is further increased according to the advertiser's request. Lastly, setting a budget threshold of 75,000, we can use Algorithm \ref{alg:algorithm_sub} to force an early exit of the advertiser after convergence to a stable bid that is not sustainable.

Increasing the cost function parameter to $k=1.4$, we predict the time to convergence to be at most $t^* = 19$ and in simulation (row 2 of Figure \ref{fig:bid_dynamics}) we see convergence in 18 iterations with 99.87\% of the budget spent. Thus, our theoretical results for the stable region are further validated in simulation.  Moreover, we once again exhibit convergence to increasing bid values in the nonuniform budget pacing context, as well as successful campaign termination by Algorithm \ref{alg:algorithm_sub}.

We lastly show the interesting dynamics that arise when $k \geq 2$. At $k=2$ we have the birth of 2-cycles in the dynamics, never converging on a singular fixed point but remaining oscillatory (as discussed in Section \ref{sec.opt_insta}). As we further increase the parameter, higher order cycles emerge until the system devolves into aperiodic behavior, more formally known as chaos \cite{strogatz:1994}. The results of this parameter tuning align with the bifurcations exhibited by the logistic map ($f(x) = rx(1-x)$) \cite{strogatz:1994}. The last row of Figure \ref{fig:bid_dynamics} plots the bid dynamics for $k = 2.3$ using the three algorithms. We see that in all instances, the bid value oscillates about a bid strategy until the campaign's conclusion or budget exhaustion. However, even in this unstable parameter region, our algorithm effectively mitigates this oscillatory behavior and ensures that the advertisers stays in the ad campaign for its total duration.

\subsection{First Price Auction}

To further demonstrate the efficacy of our budget pacing algorithm, we invoke an agent running Algorithm \ref{alg:algorithm} within a first-price auction mechanism as pitted against advertisers from the Yahoo! real-time bidding dataset. This dataset contains the bids over time of all advertisers participating in Yahoo! Search Marketing auctions for the top 1000 search queries during the period from June 15, 2002, to June 14, 2003 (segregated into 15 minute time periods for each day). We simulate the efficiency of our procedure by selecting a day at random from the dataset and selecting bid values for arriving impressions within each 15 minute window via Algorithm \ref{alg:algorithm}. On each impression in this period, we implement a first-price auction to decide on a winner and incurred cost (ie. the highest bidder wins the impression and pays exactly what they bid). Thus, for the submitted bid value by the algorithm, we observe cost $c_t$ (proportional to the number of auctions won) for a given 15 minute period and update the bid value for the next period in an effort to uniformly pace the budget expenditure. Figure \ref{fig:budget} plots the proportion of budget spent throughout the day as compared to a perfectly uniform pacing.  

\begin{figure}[h]
	\center{\includegraphics[width=0.78\linewidth]{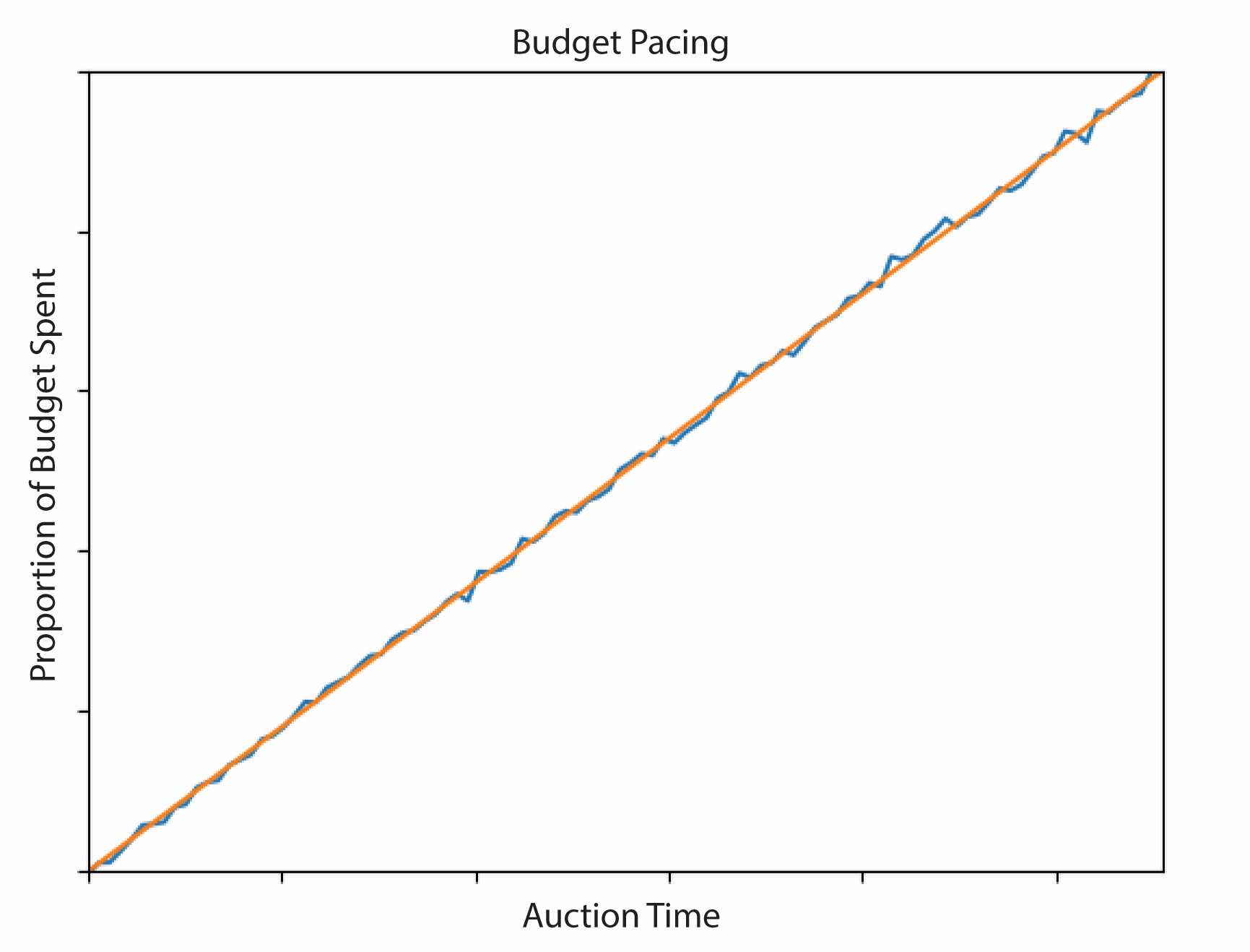}}
    \caption{\label{fig:budget} Proportion of the total budget spent as a function of time for an advertiser running Algorithm \ref{alg:algorithm} in a repeated first-price auction format. The orange curve indicates uniformly paced budget and blue indicates the algorithms performance.}
\end{figure}

\begin{figure}[h]
	\center{\includegraphics[width=0.78\linewidth]{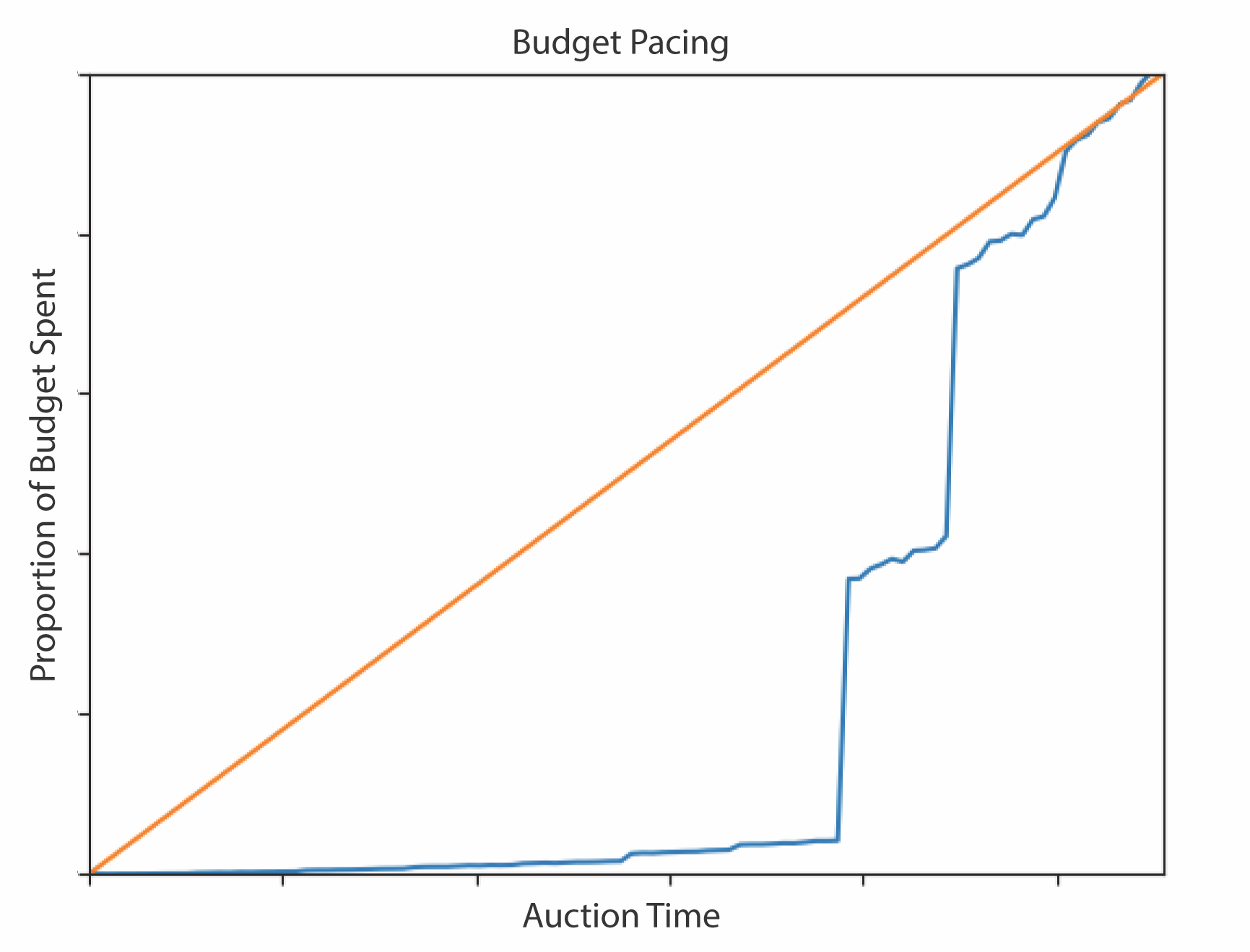}}
    \caption{\label{fig:budget_nonu} Proportion of the total budget spent as a function of time for an advertiser running Algorithm \ref{alg:algorithm_nonu} in a repeated first-price auction format. The orange curve indicates nonuniformly paced budget and blue indicates the algorithms \emph{nonuniform }performance.}
\end{figure}

We see that -- even though the spent amount function is changing in each time period by nature of the first-price auction, number of bidders, and impression values -- our algorithm successfully adjusts its bidding throughout the day and maintains a near uniform pacing of the allotted daily budget.

Intuitively, the algorithm is here increasing its bid value for a 15 minute window when it observes a cost for the prior period that indicates too few auctions were won. In a symmetric manner, when bidding too high and more spending occurs than desired, the algorithm lowers its bid value to win less of the total auctions for a period. As a result, our very simplistic algorithm manages to successfully pace its budget uniformly in the given realistic framework.

We further implement Algorithm \ref{alg:algorithm_nonu} within the repeated first price auction setting, where we invoke multipliers that force minimal spending at the start of the day which slowly increases until the end of day. Such a pacing may be used for advertisers who are hoping to capture their market after the workday or those who heavily use a given platform in the evening. Figure \ref{fig:budget_nonu} plots the proportion of budget spent throughout the day for this strategy as compared to a perfectly uniform pacing.

\section{Conclusion}
\label{sec.conc}
Bid optimization has become the central issue of budget pacing in online RTB problems. Our work has presented and analyzed a natural learning algorithm that ensures fast convergence times to optimal bids for uniform pacing of an advertiser's budget. The algorithm is additionally robust to fluctuations in the underlying mechanism of the display ad sales, in the form of a latent cost function. The efficiency of our algorithm is demonstrated both analytically and numerically through simulations on real data, demonstrating the feasibility of its implementation in real world systems. Additionally, we highlight that the algorithm does not rely on a synthesis of neural learning systems or more complex parameter estimation: it simply learns optimal behavior by comparing performance in the current iteration to the previous \cite{Perlich_2012}. 

\bibliographystyle{ACM-Reference-Format}
\bibliography{aamas}


\begin{thebibliography}{24}


\ifx \showCODEN    \undefined \def \showCODEN     #1{\unskip}     \fi
\ifx \showDOI      \undefined \def \showDOI       #1{#1}\fi
\ifx \showISBNx    \undefined \def \showISBNx     #1{\unskip}     \fi
\ifx \showISBNxiii \undefined \def \showISBNxiii  #1{\unskip}     \fi
\ifx \showISSN     \undefined \def \showISSN      #1{\unskip}     \fi
\ifx \showLCCN     \undefined \def \showLCCN      #1{\unskip}     \fi
\ifx \shownote     \undefined \def \shownote      #1{#1}          \fi
\ifx \showarticletitle \undefined \def \showarticletitle #1{#1}   \fi
\ifx \showURL      \undefined \def \showURL       {\relax}        \fi
\providecommand\bibfield[2]{#2}
\providecommand\bibinfo[2]{#2}
\providecommand\natexlab[1]{#1}
\providecommand\showeprint[2][]{arXiv:#2}

\bibitem[\protect\citeauthoryear{Agarwal, Ghosh, Wei, and You}{Agarwal
  et~al\mbox{.}}{2014}]%
        {Agarwal_2014}
\bibfield{author}{\bibinfo{person}{Deepak Agarwal}, \bibinfo{person}{Souvik
  Ghosh}, \bibinfo{person}{Kai Wei}, {and} \bibinfo{person}{Siyu You}.}
  \bibinfo{year}{2014}\natexlab{}.
\newblock \showarticletitle{Budget Pacing for Targeted Online Advertisements at
  LinkedIn}. In \bibinfo{booktitle}{\emph{Proceedings of the 20th ACM SIGKDD
  International Conference on Knowledge Discovery and Data Mining}} (New York,
  New York, USA) \emph{(\bibinfo{series}{KDD '14})}.
  \bibinfo{publisher}{Association for Computing Machinery},
  \bibinfo{address}{New York, NY, USA}, \bibinfo{pages}{1613–1619}.
\newblock
\showISBNx{9781450329569}
\urldef\tempurl%
\url{https://doi.org/10.1145/2623330.2623366}
\showDOI{\tempurl}


\bibitem[\protect\citeauthoryear{Andrzej~Grana}{Andrzej~Grana}{2003}]%
        {fixpt_thrm}
\bibfield{author}{\bibinfo{person}{James~Dugundji Andrzej~Grana}.}
  \bibinfo{year}{2003}\natexlab{}.
\newblock \bibinfo{booktitle}{\emph{Fixed Point Theory}}.
\newblock \bibinfo{publisher}{Springer}, \bibinfo{address}{New York}.
\newblock


\bibitem[\protect\citeauthoryear{Avadhanula, Colini-Baldeschi, Leonardi,
  Sankararaman, and Schrijvers}{Avadhanula et~al\mbox{.}}{2021}]%
        {Avadhanula_2021}
\bibfield{author}{\bibinfo{person}{Vashist Avadhanula},
  \bibinfo{person}{Riccardo Colini-Baldeschi}, \bibinfo{person}{S. Leonardi},
  \bibinfo{person}{Karthik~Abinav Sankararaman}, {and} \bibinfo{person}{Okke
  Schrijvers}.} \bibinfo{year}{2021}\natexlab{}.
\newblock \showarticletitle{Stochastic bandits for multi-platform budget
  optimization in online advertising}.
\newblock \bibinfo{journal}{\emph{Proceedings of the Web Conference 2021}}
  (\bibinfo{year}{2021}).
\newblock


\bibitem[\protect\citeauthoryear{Balseiro and Gur}{Balseiro and Gur}{2017}]%
        {Balseiro_2019}
\bibfield{author}{\bibinfo{person}{S. Balseiro} {and} \bibinfo{person}{Y.
  Gur}.} \bibinfo{year}{2017}\natexlab{}.
\newblock \showarticletitle{Learning in Repeated Auctions with Budgets: Regret
  Minimization and Equilibrium}.
\newblock \bibinfo{journal}{\emph{Proceedings of the 2017 ACM Conference on
  Economics and Computation}} (\bibinfo{year}{2017}).
\newblock


\bibitem[\protect\citeauthoryear{Balseiro, Kim, Mahdian, and Mirrokni}{Balseiro
  et~al\mbox{.}}{2021}]%
        {Balseiro_2021}
\bibfield{author}{\bibinfo{person}{S. Balseiro}, \bibinfo{person}{A. Kim},
  \bibinfo{person}{Mohammad Mahdian}, {and} \bibinfo{person}{V. Mirrokni}.}
  \bibinfo{year}{2021}\natexlab{}.
\newblock \showarticletitle{Budget-Management Strategies in Repeated Auctions}.
\newblock \bibinfo{journal}{\emph{Oper. Res.}}  \bibinfo{volume}{69}
  (\bibinfo{year}{2021}), \bibinfo{pages}{859--876}.
\newblock


\bibitem[\protect\citeauthoryear{Borgs, Chayes, Immorlica, Jain, Etesami, and
  Mahdian}{Borgs et~al\mbox{.}}{2007}]%
        {borgs_2007}
\bibfield{author}{\bibinfo{person}{C. Borgs}, \bibinfo{person}{J. Chayes},
  \bibinfo{person}{Nicole Immorlica}, \bibinfo{person}{K. Jain},
  \bibinfo{person}{O. Etesami}, {and} \bibinfo{person}{Mohammad Mahdian}.}
  \bibinfo{year}{2007}\natexlab{}.
\newblock \showarticletitle{Dynamics of bid optimization in online
  advertisement auctions}. In \bibinfo{booktitle}{\emph{WWW '07}}.
\newblock


\bibitem[\protect\citeauthoryear{Cai, Ren, Zhang, Malialis, Wang, Yu, and
  Guo}{Cai et~al\mbox{.}}{2017}]%
        {Cai_2017}
\bibfield{author}{\bibinfo{person}{Han Cai}, \bibinfo{person}{Kan Ren},
  \bibinfo{person}{Weinan Zhang}, \bibinfo{person}{Kleanthis Malialis},
  \bibinfo{person}{Jun Wang}, \bibinfo{person}{Yong Yu}, {and}
  \bibinfo{person}{Defeng Guo}.} \bibinfo{year}{2017}\natexlab{}.
\newblock \showarticletitle{Real-Time Bidding by Reinforcement Learning in
  Display Advertising}.
\newblock \bibinfo{journal}{\emph{CoRR}}  \bibinfo{volume}{abs/1701.02490}
  (\bibinfo{year}{2017}).
\newblock
\showeprint[arxiv]{1701.02490}
\urldef\tempurl%
\url{http://arxiv.org/abs/1701.02490}
\showURL{%
\tempurl}


\bibitem[\protect\citeauthoryear{Celli, Colini-Baldeschi, Kroer, and
  Sodomka}{Celli et~al\mbox{.}}{2021}]%
        {Celli_2021}
\bibfield{author}{\bibinfo{person}{A. Celli}, \bibinfo{person}{Riccardo
  Colini-Baldeschi}, \bibinfo{person}{Christian Kroer}, {and}
  \bibinfo{person}{Eric Sodomka}.} \bibinfo{year}{2021}\natexlab{}.
\newblock \showarticletitle{The Parity Ray Regularizer for Pacing in Auction
  Markets}.
\newblock \bibinfo{journal}{\emph{ArXiv}}  \bibinfo{volume}{abs/2106.09503}
  (\bibinfo{year}{2021}).
\newblock


\bibitem[\protect\citeauthoryear{Conitzer, Kroer, Sodomka, and
  Stier-Moses}{Conitzer et~al\mbox{.}}{2018}]%
        {Conitzer_2021}
\bibfield{author}{\bibinfo{person}{Vincent Conitzer},
  \bibinfo{person}{Christian Kroer}, \bibinfo{person}{Eric Sodomka}, {and}
  \bibinfo{person}{N. Stier-Moses}.} \bibinfo{year}{2018}\natexlab{}.
\newblock \showarticletitle{Multiplicative Pacing Equilibria in Auction
  Markets}. In \bibinfo{booktitle}{\emph{WINE}}.
\newblock


\bibitem[\protect\citeauthoryear{Feng, Podimata, and Syrgkanis}{Feng
  et~al\mbox{.}}{2018}]%
        {Feng_2018}
\bibfield{author}{\bibinfo{person}{Zhe Feng}, \bibinfo{person}{Chara Podimata},
  {and} \bibinfo{person}{Vasilis Syrgkanis}.} \bibinfo{year}{2018}\natexlab{}.
\newblock \showarticletitle{Learning to Bid Without Knowing your Value}.
\newblock \bibinfo{journal}{\emph{Proceedings of the 2018 ACM Conference on
  Economics and Computation}} (\bibinfo{year}{2018}).
\newblock


\bibitem[\protect\citeauthoryear{Fernandez-Tapia}{Fernandez-Tapia}{2015}]%
        {tapia_2015}
\bibfield{author}{\bibinfo{person}{Joaquin Fernandez-Tapia}.}
  \bibinfo{year}{2015}\natexlab{}.
\newblock \showarticletitle{An analytical solution to the budget-pacing problem
  in programmatic advertising}.
\newblock \bibinfo{journal}{\emph{Journal of Information and Optimization
  Sciences}}  \bibinfo{volume}{40} (\bibinfo{date}{November}
  \bibinfo{year}{2015}).
\newblock
\urldef\tempurl%
\url{https://doi.org/10.1080/02522667.2017.1303946}
\showDOI{\tempurl}


\bibitem[\protect\citeauthoryear{Geyik, Saxena, and Dasdan}{Geyik
  et~al\mbox{.}}{2015}]%
        {Geyik_2015}
\bibfield{author}{\bibinfo{person}{Sahin~Cem Geyik}, \bibinfo{person}{Abhishek
  Saxena}, {and} \bibinfo{person}{Ali Dasdan}.}
  \bibinfo{year}{2015}\natexlab{}.
\newblock \bibinfo{title}{Multi-Touch Attribution Based Budget Allocation in
  Online Advertising}.
\newblock
\newblock
\showeprint[arxiv]{1502.06657}~[cs.AI]


\bibitem[\protect\citeauthoryear{Jauvion and Grislain}{Jauvion and
  Grislain}{2020}]%
        {Jauvion_2020}
\bibfield{author}{\bibinfo{person}{Gr{\'{e}}goire Jauvion} {and}
  \bibinfo{person}{Nicolas Grislain}.} \bibinfo{year}{2020}\natexlab{}.
\newblock \showarticletitle{Optimal Allocation of Real-Time-Bidding and Direct
  Campaigns}.
\newblock \bibinfo{journal}{\emph{CoRR}}  \bibinfo{volume}{abs/2006.07070}
  (\bibinfo{year}{2020}).
\newblock
\showeprint[arxiv]{2006.07070}
\urldef\tempurl%
\url{https://arxiv.org/abs/2006.07070}
\showURL{%
\tempurl}


\bibitem[\protect\citeauthoryear{Karaca, Sessa, Leidi, and Kamgarpour}{Karaca
  et~al\mbox{.}}{2019}]%
        {Karaca_2019}
\bibfield{author}{\bibinfo{person}{Orcun Karaca},
  \bibinfo{person}{Pier~Giuseppe Sessa}, \bibinfo{person}{A. Leidi}, {and}
  \bibinfo{person}{M. Kamgarpour}.} \bibinfo{year}{2019}\natexlab{}.
\newblock \showarticletitle{No-Regret Learning from Partially Observed Data in
  Repeated Auctions}.
\newblock \bibinfo{journal}{\emph{ArXiv}}  \bibinfo{volume}{abs/1912.09905}
  (\bibinfo{year}{2019}).
\newblock


\bibitem[\protect\citeauthoryear{Lee, Jalali, and Dasdan}{Lee
  et~al\mbox{.}}{2013}]%
        {lee2013}
\bibfield{author}{\bibinfo{person}{Kuang{-}Chih Lee}, \bibinfo{person}{Ali
  Jalali}, {and} \bibinfo{person}{Ali Dasdan}.}
  \bibinfo{year}{2013}\natexlab{}.
\newblock \showarticletitle{Real Time Bid Optimization with Smooth Budget
  Delivery in Online Advertising}.
\newblock \bibinfo{journal}{\emph{CoRR}}  \bibinfo{volume}{abs/1305.3011}
  (\bibinfo{year}{2013}).
\newblock
\showeprint[arxiv]{1305.3011}
\urldef\tempurl%
\url{http://arxiv.org/abs/1305.3011}
\showURL{%
\tempurl}


\bibitem[\protect\citeauthoryear{Lobos, Grigas, Wen, and chih Lee}{Lobos
  et~al\mbox{.}}{2018}]%
        {lobos2018optimal}
\bibfield{author}{\bibinfo{person}{Alfonso Lobos}, \bibinfo{person}{Paul
  Grigas}, \bibinfo{person}{Zheng Wen}, {and} \bibinfo{person}{Kuang chih
  Lee}.} \bibinfo{year}{2018}\natexlab{}.
\newblock \bibinfo{title}{Optimal Bidding, Allocation and Budget Spending for a
  Demand Side Platform Under Many Auction Types}.
\newblock
\newblock
\showeprint[arxiv]{1805.11645}~[math.OC]


\bibitem[\protect\citeauthoryear{Nedelec, Karoui, and Perchet}{Nedelec
  et~al\mbox{.}}{2019}]%
        {Nedelec_2019}
\bibfield{author}{\bibinfo{person}{Thomas Nedelec},
  \bibinfo{person}{Noureddine~El Karoui}, {and} \bibinfo{person}{Vianney
  Perchet}.} \bibinfo{year}{2019}\natexlab{}.
\newblock \showarticletitle{Learning to bid in revenue-maximizing auctions}.
\newblock \bibinfo{journal}{\emph{ArXiv}}  \bibinfo{volume}{abs/1902.10427}
  (\bibinfo{year}{2019}).
\newblock


\bibitem[\protect\citeauthoryear{Perlich, Dalessandro, Hook, Stitelman, Raeder,
  and Provost}{Perlich et~al\mbox{.}}{2012}]%
        {Perlich_2012}
\bibfield{author}{\bibinfo{person}{Claudia Perlich}, \bibinfo{person}{Brian
  Dalessandro}, \bibinfo{person}{Rod Hook}, \bibinfo{person}{Ori Stitelman},
  \bibinfo{person}{Troy Raeder}, {and} \bibinfo{person}{Foster Provost}.}
  \bibinfo{year}{2012}\natexlab{}.
\newblock \showarticletitle{Bid Optimizing and Inventory Scoring in Targeted
  Online Advertising}. In \bibinfo{booktitle}{\emph{Proceedings of the 18th ACM
  SIGKDD International Conference on Knowledge Discovery and Data Mining}}
  (Beijing, China) \emph{(\bibinfo{series}{KDD '12})}.
  \bibinfo{publisher}{Association for Computing Machinery},
  \bibinfo{address}{New York, NY, USA}, \bibinfo{pages}{804–812}.
\newblock
\showISBNx{9781450314626}
\urldef\tempurl%
\url{https://doi.org/10.1145/2339530.2339655}
\showDOI{\tempurl}


\bibitem[\protect\citeauthoryear{Rudin}{Rudin}{1976}]%
        {rudin}
\bibfield{author}{\bibinfo{person}{Walter Rudin}.}
  \bibinfo{year}{1976}\natexlab{}.
\newblock \bibinfo{booktitle}{\emph{Principles of mathematical analysis /
  Walter Rudin} (\bibinfo{edition}{3d ed.} ed.)}.
\newblock \bibinfo{publisher}{McGraw-Hill New York}.
\newblock
\showISBNx{007054235}
\urldef\tempurl%
\url{http://www.loc.gov/catdir/toc/mh031/75017903.html}
\showURL{%
\tempurl}


\bibitem[\protect\citeauthoryear{Strogatz}{Strogatz}{2000}]%
        {strogatz:1994}
\bibfield{author}{\bibinfo{person}{Steven~H. Strogatz}.}
  \bibinfo{year}{2000}\natexlab{}.
\newblock \bibinfo{booktitle}{\emph{{Nonlinear Dynamics and Chaos: With
  Applications to Physics, Biology, Chemistry and Engineering}}}.
\newblock \bibinfo{publisher}{Westview Press}.
\newblock


\bibitem[\protect\citeauthoryear{Tunuguntla and Hoban}{Tunuguntla and
  Hoban}{2021}]%
        {Tunuguntla_2021}
\bibfield{author}{\bibinfo{person}{Srinivas Tunuguntla} {and}
  \bibinfo{person}{Paul~R. Hoban}.} \bibinfo{year}{2021}\natexlab{}.
\newblock \showarticletitle{A Near-Optimal Bidding Strategy for Real-Time
  Display Advertising Auctions}.
\newblock \bibinfo{journal}{\emph{Journal of Marketing Research}}
  \bibinfo{volume}{58} (\bibinfo{year}{2021}), \bibinfo{pages}{p. 1--21}.
\newblock
\urldef\tempurl%
\url{https://doi.org/10.1177/0022243720968547}
\showDOI{\tempurl}
\showeprint{https://doi.org/10.1177/0022243720968547}


\bibitem[\protect\citeauthoryear{Weed, Perchet, and Rigollet}{Weed
  et~al\mbox{.}}{2016}]%
        {Weed_2016}
\bibfield{author}{\bibinfo{person}{J. Weed}, \bibinfo{person}{Vianney Perchet},
  {and} \bibinfo{person}{P. Rigollet}.} \bibinfo{year}{2016}\natexlab{}.
\newblock \showarticletitle{Online learning in repeated auctions}. In
  \bibinfo{booktitle}{\emph{COLT}}.
\newblock


\bibitem[\protect\citeauthoryear{Xu, Lee, Li, Qi, and Lu}{Xu
  et~al\mbox{.}}{2015}]%
        {Xu_2015}
\bibfield{author}{\bibinfo{person}{Jian Xu}, \bibinfo{person}{Kuang-chih Lee},
  \bibinfo{person}{Wentong Li}, \bibinfo{person}{Hang Qi}, {and}
  \bibinfo{person}{Quan Lu}.} \bibinfo{year}{2015}\natexlab{}.
\newblock \showarticletitle{Smart Pacing for Effective Online Ad Campaign
  Optimization}.
\newblock \bibinfo{journal}{\emph{Proceedings of the 21th ACM SIGKDD
  International Conference on Knowledge Discovery and Data Mining}}
  (\bibinfo{date}{Aug} \bibinfo{year}{2015}).
\newblock
\showISBNx{9781450336642}
\urldef\tempurl%
\url{https://doi.org/10.1145/2783258.2788615}
\showDOI{\tempurl}


\bibitem[\protect\citeauthoryear{Zhang, Yuan, and Wang}{Zhang
  et~al\mbox{.}}{2014}]%
        {Zhang_2014}
\bibfield{author}{\bibinfo{person}{Weinan Zhang}, \bibinfo{person}{Shuai Yuan},
  {and} \bibinfo{person}{Jun Wang}.} \bibinfo{year}{2014}\natexlab{}.
\newblock \showarticletitle{Optimal Real-Time Bidding for Display Advertising}.
  In \bibinfo{booktitle}{\emph{Proceedings of the 20th ACM SIGKDD International
  Conference on Knowledge Discovery and Data Mining}} (New York, New York, USA)
  \emph{(\bibinfo{series}{KDD '14})}. \bibinfo{publisher}{Association for
  Computing Machinery}, \bibinfo{address}{New York, NY, USA},
  \bibinfo{pages}{1077–1086}.
\newblock
\showISBNx{9781450329569}
\urldef\tempurl%
\url{https://doi.org/10.1145/2623330.2623633}
\showDOI{\tempurl}


\end{thebibliography}

\newpage
\appendix
\label{sec:appendix}
\onecolumn
\section{Omitted Proofs}
\subsection{Proof of Lemma \ref{lemma:bounded}} We proceed by bounding the cost function. For $0 < b < 1$, we bound the function using the minimal exponent
\begin{align*}
	c_1b^{k_m} < c_1b^{k_1} + c_2b^{k_2} + ... + c_mb^{k_m} < \tilde C b^{k_m}
\end{align*}
where $\tilde C := \sum_i^m |c_i|$ and identically for $b \geq 1$, we bound the function with the maximal exponent
\begin{align*}
	c_1b^{k_1} < c_1b^{k_1} + c_2b^{k_2} + ... + c_mb^{k_m} < \tilde C b^{k_1}
\end{align*}
For convenience, let $s(b) := c_1b^{k_1} + c_2b^{k_2} + ... + c_mb^{k_m}$ and without loss of generality assume $b \geq 1$. Using the above we demonstrate the following bounds on Algorithm 1 for the general cost function
\begin{align*}
	c_1b^{k_m} &< s(b) < \tilde C b^{k_m} \\
	\frac{1}{c_1b^{k_m}} &> \frac{1}{s(b)} > \frac{1}{\tilde C b^{k_m}} \\
	b \cdot \left(\frac{B}{Tc_1b^{k_m}} - \frac{1}{T} \right) &> b \cdot \left(\frac{B}{Ts(b)} - \frac{1}{T} \right) \\
	&> b \cdot \left(\frac{B}{T\tilde C b^{k_m}} - \frac{1}{T} \right)
\end{align*}
which is equivalent to the output of Algorithm 1 with $t = 0$. Thus, we have the algorithm bound
\begin{align*}
	&\text{Algorithm 1}(B,T,0,b_0,c_1b_0^{k_m}) \\
	&> \text{Algorithm 1}(B,T,0,b_0,s(b_0)) \\
	&> \text{Algorithm 1}(B,T,0,b_0,\tilde C b_0^{k_m})
\end{align*}
Lastly, manipulating the result from Banach in Lemma 3.5 we have the convergence time bound
\begin{align*}
	t \geq \ln\left(\frac{(1-L) |b_t - b^*|}{|b_1 - b_0|}\right) / \ln(L)
\end{align*}
substituting the output of Algorithm 1 for $b_1$, we thus define the convergence time bound for a cost function $c$ as
\begin{align*}
	 \tau_{s(b)} := \ln\left(\frac{(1-L) |b_t - b^*|}{|\text{Algorithm 1}(B,T,0,b_0,s(b_0)) - b_0|}\right) / \ln(L) \tag{A1.1} \label{eq:tau}
\end{align*}
Now combining with the algorithm bounds above gives 
\begin{align*}
	&\frac{(1-L) |b_t - b^*|}{|\text{Algorithm 1}(B,T,0,b_0,c_1b_0^{k_m}) - b_0|} \\
	&< \frac{(1-L) |b_t - b^*|}{|\text{Algorithm 1}(B,T,0,b_0,s(b_0)) - b_0|} \\
	&< \frac{(1-L) |b_t - b^*|}{|\text{Algorithm 1}(B,T,0,b_0,\tilde C b_0^{k_m}) - b_0|}
\end{align*}
which using definition \eqref{eq:tau} equates to
\begin{align*}
	L^{\tau_{c_1b^{k_m}}} < L^{\tau_{s}} < L^{\tau_{\tilde C b^{k_m}}} \Rightarrow \tau_{c_1b^{k_m}} < \tau_{s(b)} < \tau_{\tilde C b^{k_m}}
\end{align*}
bounding the general cost function's convergence time with monomials.

\subsection{Proof of Theorem \ref{thm:chaos}} A 2-cycle for a function $f(x)$ exists if and only if there are two points $x_1,x_2$ such that $f(x_1) = x_2$ and $f(x_2) = x_1$. Equivalently, these points must satisfy $f(f(x_1)) = x_1$ and $f(f(x_2)) = x_2$. We proceed by finding the points that satisfy the latter condition, in our context being the points which are fixed points when applying Algorithm 1 twice. Solving this relation amounts to the following problem
\begin{align*}
	b_1 &= \text{Algorithm 1}(B,T,0,b_0,\min\{cb_0^k,M\}) \\
	&= \frac{B - \min\{cb_0^k,M\}}{\min\{cb_0^k,M\} T} \cdot b_0 \\
	b_2 &= \text{Algorithm 1}(B,T,1,b_1,(\min\{cb_0^k,M\},\min\{cb_1^k,M\})) \\
	&= \frac{B - \min\{cb_0^k,M\} - \min\{cb_1^k,M\}}{\min\{cb_1^k,M\} (T-1)} \cdot b_1 \\
	= &\frac{B - \min\{cb_0^k,M\} - \min\{c\left(\frac{B - \min\{cb_0^k,M\}}{\min\{cb_0^k,M\} T} \cdot b_0\right)^k,M\}}{\min\{c\left(\frac{B - \min\{cb_0^k,M\}}{\min\{cb_0^k,M\} T} \cdot b_0\right)^k,M\} (T-1)} \\
	&\cdot \frac{B - \min\{cb_0^k,M\}}{\min\{cb_0^k,M\} T} \cdot b_0 \tag{A2.1} \label{eq:2point}
\end{align*}
We now consider three unique cases for the value of the minimizations:
\vspace{2mm} \\
\underline{Case 1:} $\min\{cb_0^k,M\} = cb_0^k, \min\{cb_1^k,M\} = cb_1^k$ \\
Setting \eqref{eq:2point} $= b_0$ and solving gives us
\begin{align*}
	b_0 = \left(\frac{B}{c(T+1)}\right)^{1/k}
\end{align*}
which is the original fixed point )\ref{eq:fpt}) from Section 3.3 that was proved to now be unstable for $k > 2$. As a result, the value of the bid will be repelled from this point to the two points of our cycle. It is important to note that any fixed point of a function $f$ satisfies $f(x) = x$ so we automatically have $f(f(x)) = f(x) = x$. We thus need not consider this trivial point in the 2 cycle, as we are interested in the situation where $f(f(x)) = x$ and $f(x) \neq x$.
\vspace{2mm} \\
\underline{Case 2:} $\min\{cb_0^k,M\} = cb_0^k, \min\{cb_1^k,M\} = M$ \\
Once again solving \eqref{eq:2point} $= b_0$ produces the following point
\begin{align*}
	b_0 = 2^{-1/k}(\sqrt{(-2B - MT^2 + MT + M)^2 + 4B(M-B)} \\
	+ 2B + MT^2 - MT - M)^{1/k}
\end{align*}
for which we need to verify $b_0 \in \mathbb{R}$. Real solutions exist for
\begin{gather*}
	\sqrt{(-2B - MT^2 + MT + M)^2 + 4B(M-B)} + 2B + ...\\
	 MT^2 - MT - M > 0 \\
	\Rightarrow (-2B - MT^2 + MT + M)^2 + 4B(M-B) \\
	 > (M + MT - 2B - MT^2)^2 \\
	M > B
\end{gather*}
however this last inequality on M does not hold. Therefore $b_0 \in \mathbb{C}$ in this case and thus this is not one of our period two points since our bid values are constrained to the real numbers.
\vspace{2mm} \\
\underline{Case 3:} $\min\{cb_0^k,M\} = M, \min\{cb_1^k,M\} = cb_1^k$ \\
Solving \eqref{eq:2point} $= b_0$ once again yields
\begin{align*}
	b_0 = \left( \frac{(B-M)^2}{c(B-M-MT+MT^2)}\right)^{1/k}
\end{align*}
for which we once again need to verify $b_0 \in \mathbb{R}$. It is clear that $(B-M)^2 > 0$, so we need only check
\begin{gather*}
	c(B - M - MT + MT^2) > 0 \\
	cM(T^2-T-1) > B \\
	M > \frac{B}{c(T^2-T-1)}
\end{gather*} 
For $T > 1$, $\frac{B}{cT} > \frac{B}{c(T^2 - T - 1)}$. Combining this with fact that $M > \frac{B}{cT}$, verifies the above inequality and thus $b_0 \in \mathbb{R}$. Therefore, this is one of the points in our 2-period cycle. Lastly, to find the second point of this cycle, we simply apply the algorithm again. Let $b_-$ denote the first point. We once again have two subcases which depend on the value of $c$. For $\min\{cb_-^k,M\} = cb_-^k$
\begin{align*}
	b_+ &= \text{Algorithm 1}(B,T,0,b_-,\min\{b_-^k,M\}) \\
	&= \frac{B - \min\{b_-^k,M\}}{\min\{b_-^k,M\}T} \cdot b_- \\
	&= \frac{B - b_-^k}{b_-^kT} \cdot b_- \\
	&= \frac{MT}{B-M}\left(\frac{(B-M)^2}{c(B-M-MT+MT^2)}\right)^{1/k} \\
	&= \frac{MT}{B-M} \cdot b_-
\end{align*}
and for $\min\{cb_-^k,M\} = M$
\begin{align*}
	b_+ = \frac{B - M}{MT} \cdot b_- \\
\end{align*}
Therefore we have our two cycle producing oscillatory behavior between $b_+$ and $b_-$.


\end{document}